\documentclass[11pt,english]{article}
\usepackage{courier}
\usepackage[T1]{fontenc}
\usepackage[latin9]{inputenc}
\usepackage[letterpaper]{geometry}
\usepackage{textcomp}
\usepackage{amsthm}
\usepackage{amsmath}
\usepackage{amssymb}
\usepackage{esint}

\makeatletter

\newcommand{\noun}[1]{\textsc{#1}}

\numberwithin{equation}{section}
\numberwithin{figure}{section}
  \theoremstyle{plain}
  \newtheorem*{thm*}{Theorem}
\theoremstyle{plain}
\newtheorem{thm}{Theorem}
  \theoremstyle{plain}
  \newtheorem{lem}[thm]{Lemma}
  \theoremstyle{plain}
  \newtheorem{cor}[thm]{Corollary}
  \theoremstyle{remark}
  \newtheorem{rem}[thm]{Remark}

\makeatother

\makeatother

\usepackage{babel}

\begin{document}

\title{On Villani's Conjecture Concerning Entropy Production for the Kac Master Equation}

\date{}
\author{Amit Einav%
\thanks{The work presented in this paper was supported by U.S. National Science
Foundation grant DMS-0901304%
}\\
 Department of Mathematics\\
 Georgia Institute of Technology}
\maketitle
\begin{abstract}
In this paper we take an idea presented in recent paper by Carlen,
Carvalho, Le Roux, Loss, and Villani (\cite{key-7}) and push it one
step forward to find an exact estimation on the entropy production.
The new estimation essentially proves that Villani's conjecture is
correct, or more precisely that a much worse bound to the entropy
production is impossible in the general case. %
\end{abstract}
\tableofcontents{}

\section{Introduction\label{sec:Introduction}}

In his 1956 paper on the Foundations of Kinetic Theory (\cite{key-9}),
Mark Kac proposed a probabilistic model describing a system of $N$
one dimensional, randomly colliding particles. The description is
given by Kac's Master Equation \begin{equation}
\frac{\partial\psi}{\partial t}\left(v_{1},\dots,v_{N},t\right)=-N(I-Q)\psi\left(v_{1},\dots,v_{N},t\right)\label{master}\end{equation}
 where \[
Q\phi\left(v_{1},\dots,v_{N}\right)=\frac{1}{2\pi}\cdot\frac{1}{\left(\begin{array}{c}
N\\
2\end{array}\right)}\sum_{i<j}\int_{0}^{2\pi}\phi\left(R_{i,j}(\vartheta)\left(v_{1},\dots,v_{N}\right)\right)d\vartheta\]
 with \[
R_{i,j}(\vartheta)\left(v_{1},\dots,v_{N}\right)=\left(v_{1},\dots v_{i}(\vartheta),\dots,v_{j}(\vartheta),\dots,v_{N}\right)\]
 \[
v_{i}(\vartheta)=v_{i}\cos\vartheta+v_{j}\sin\vartheta,\,\, v_{j}(\vartheta)=-v_{i}\sin\vartheta+v_{j}\cos\vartheta\ .\]
 The function $\psi(v_{1},\dots,v_{N},t)$ is a probability distribution
on the energy sphere and it is formally given by \[
\psi(\cdot,t)=e^{-N(I-Q)t}\psi_{0}\]
 for some initial condition $\psi_{0}$. In the same paper, Kac introduced
the notion of chaotic sequences (although he did not call it that
way) and showed that this notion is preserved under the time evolution.
This property is now called Propagation of Chaos. Kac went further
and showed in fact that single particle marginal of the evolved density
is a solution of the model Boltzmann equation \[
\frac{\partial f}{\partial t}(v,t)=\frac{1}{2\pi}\int_{\mathbb{R}}d\omega\int_{0}^{2\pi}d\vartheta\left(f\left(v\cos\vartheta+\omega\sin\vartheta,t\right)f\left(-v\sin\vartheta+\omega\cos\vartheta,t\right)-f(v,t)f(\omega,t)\right)\]
 and thus giving a cogent derivation of the spatially homogeneous
Boltzmann equation. For a detailed review the reader may consult \cite{key-7}.

The equation (\ref{master}), or rather the operator, is a bounded
self-adjoint operator in the space $L^{2}\left(\mathbb{S}^{N-1}(\sqrt{N}),d\sigma^{N}\right)$
where $d\sigma^{N}$ is the normalized uniform measure on the sphere.
It is fairly easy to see that the time evolution defined by (\ref{master})
is ergodic, i.e., the solution will approach the function $\psi=1$
as $t\to\infty$. By the spectral theorem, the rate of approach to
the constant function in the sense of $L^{2}$ distance is governed
by the gap \[
\Delta_{N}=\inf\left\{ \left\langle \varphi,N(I-Q)\varphi\right\rangle \,\,:\,\,\left\langle \varphi,1\right\rangle =0,\,\,\left\langle \varphi,\varphi\right\rangle =1\right\} \]
 where the infimum is taken over all $\varphi\in L^{2}\left(\mathbb{S}^{N-1}(\sqrt{N}),d\sigma^{N}\right)$.
Kac conjectured that \[
\liminf_{N\rightarrow\infty}\Delta_{N}>0\ .\]
 The conjecture was proved to be true by Janvresse in (\cite{key-8})
and the exact value of $\Delta_{N}$ was computed by Carlen, Carvalho,
and Loss in (\cite{key-6}).

The $L^{2}$ distance is rather unsatisfactory. For any reasonable
density $\psi$, in particular a chaotic one, it is easy to see that
\[
\left\Vert \psi(v_{1},\dots,v_{N},0)\right\Vert _{L^{2}\left(\mathbb{S}^{N-1}(\sqrt{N}),d\sigma^{N}\right)}\geq C^{N}\]
 where $C>1$ and hence it would take a time of order $N$ to see
a substantial decay of the $L^{2}$. Clearly, this is not what one
considers {}``approach to equilibrium''. A more natural quantity
to use is the entropy \[
H_{N}(\psi)=\int_{\mathbb{S}^{N-1}(\sqrt{N})}\psi\log\psi\]

The crucial difference between the $L^{2}$ distance and the entropy
lies in the extensivity of the entropy, namely that if $\psi_{N}\left(v_{1},\dots,v_{N},t\right)$
satisfies $\psi_{N}\left(v_{1},\dots,v_{N},t\right)\approx\Pi_{i=1}^{N}f(v_{i},t)$
in a weak sense, i.e., chaotic (referred by Kac as 'The Boltzmann
Property') then \[
H_{N}(\psi_{N})\approx N\int_{\mathbb{R}}f(v,t)\log\left(\frac{f(v,t)}{\gamma(v)}\right)dv=NH(f(v,t)\vert\gamma(v))\]
 where $\gamma(v)$ is the normalized Gaussian.

Differentiating the entropy of a solution to the Kac Model gives the
time evolution equation:\[
\frac{\partial H_{N}(\psi_{N})}{\partial t}=\left\langle \log\psi_{N},N(I-Q)\psi_{N}\right\rangle \]

This, along with a known inequality by Csiszar, Kullback, Leibler
and Pinsker and the enxtensivity property allows us to conclude that
\[
\left\Vert \psi_{N}(v_{1},\dots,v_{N},t)d\sigma^{N}-d\sigma^{N}\right\Vert _{\mbox{Total Variation}}^{2}\leq2Ne^{-\Gamma_{N}t}H(f(v,0)\vert\gamma(v))\]
 for \[
\Gamma_{N}=\inf\frac{\left\langle \log\left(\psi_{N}\right),N(I-Q)\psi_{N}\right\rangle }{H_{N}(\psi_{N})}\]
 where the infimum is taken over all probability densities $\psi_{N}$
on $\mathbb{S}^{N-1}(\sqrt{N})$ which are symmetric in all their
components. $\Gamma_{N}$ is called the \emph{entropy production}.

The hope that there exists $C>0$ such that $\Gamma_{N}\geq C$ was
refuted in $2010$ in an paper by Carlen, Carvalho, Le Roux, Loss,
and Villani (\cite{key-7}) where the authors managed to find a sequence
of probability densities $\left\{ \phi_{N}\right\} _{N\in\mathbb{N}}$
with \begin{equation}
\limsup_{N\rightarrow\infty}\frac{\left\langle \log\left(\phi_{N}\right),N(I-Q)\phi_{N}\right\rangle }{H_{N}(\phi_{N})}=0\label{eq:Loss result}\end{equation}
 While this means that the time of convergence to equilibrium is not
of logarithm type, an exact estimation on the entropy production might
still give a better convergence rate than that of the original Kac
model.

The first step towards this goal was done in $2003$ by Villani in
(\cite{key-10}) who proved that\[
\Gamma_{N}\geq\frac{2}{N-1}\]
 Villani conjectured that\[
\Gamma_{N}=O\left(\frac{1}{N}\right)\]
 which wouldn't bode well for the approach to equilibrium in the ergodic
sense, but poses an interesting mathematical problem.

The main result of this paper is to show that Villani's conjecture
is essentially true. More precisely, we will show that
\begin{thm*}
For any $0<\beta<\frac{1}{6}$ there exists a constant $C_{\beta}$
depending only on $\beta$ such that\begin{equation}
\Gamma_{N}\leq\frac{C_{\beta}\log N}{N^{1-2\beta}}\label{eq:my result}\end{equation}

\end{thm*}
(See Theorem \ref{thm:tnropy production result} in Section \ref{sec:Entropy-Production}).

Both (\ref{eq:Loss result}) and (\ref{eq:my result}) are proved
with the same idea: creating an $N$ particle symmetric function $F_{N}$
from a one particle function $f$\[
F_{N}\left(v_{1},\dots,v_{N}\right)=\frac{\Pi_{i=1}^{N}f(v_{i})}{Z_{N}(f,\sqrt{N})}\]
 where \[
Z_{N}(f,r)=\int_{\mathbb{S}^{N-1}(r)}\Pi_{i=1}^{N}f(v_{i})d\sigma_{r}^{N}\]
 and $d\sigma_{r}^{N}$ is the uniform probability measure on $\mathbb{S}^{N-1}(r)$.
The main difference between the two proofs lies in the fact that while
in (\cite{key-7}) $f$ remains fixed, in our paper $f$ changes with
$N$ via a parameter $\delta=\delta_{N}$.

The paper is structured as follows: Section \ref{sec:About-the-Function Z}
reviews known results about the normalization function $Z_{N}(f,r)$.
Section \ref{sec:Centrel-Limit-Theorem} is our main theoretical part
of the paper, dealing with general properties that will allow us to
give an asymptotic expression to the normalization function. Section
\ref{sec:Entropy-Production} is where we prove our main result. Picking
a function which is natural to the problem at hand and using the result
of the previous sections along with some involved computation. Section
\ref{sec:Final-Remarks} contains a few last remarks and the Appendix
has some simple but very useful computation that we use throughout
the entire paper.

We'd like to conclude the introduction by thanking Michael Loss for
his helpful remarks and discussions, making this paper possible.

\section{The Function $Z_{N}(f,r)$\label{sec:About-the-Function Z}}

The key to the computation of the entropy production lies with the
normalization function $Z_{N}(f,r)$. In this short section we'll
find a simple probabilistic interpretation to it, along with a formula
that will serve us in the following sections and the final computation.
This section is a short review of known results from (\cite{key-7}). 
\begin{lem}
\label{lem: Definition of h}Let $f$ be a density function for the
real valued random variable $V$. Then the density function of the
random variable $V^{2}$ is given by \[
h(u)=\frac{f(\sqrt{u})+f(-\sqrt{u})}{2\sqrt{u}}\]
 \end{lem}
\begin{proof}
For any function $\varphi=\varphi(|x|)=\varphi(r)$ we find that\[
\mathbb{E}\varphi=\int_{0}^{\infty}\varphi(r)\cdot\left(f(r)+f(-r)\right)dr\]
 on the other hand\[
\mathbb{E}\varphi=\int_{0}^{\infty}\varphi\left(\sqrt{t}\right)h(t)dt=\int_{0}^{\infty}\varphi(r)\cdot2r\cdot h\left(r^{2}\right)dr\]
 Since $\varphi$ was arbitrary we find that\[
2r\cdot h\left(r^{2}\right)=f(r)+f(-r)\]
 and the result follows.\end{proof}
\begin{lem}
\label{lem:Definition of s}Let $V_{1},\dots,V_{N}$ be independent
real valued random variables with identical density function $f(v)$.
Then the density function for $S_{N}=\sum_{i=1}^{N}V_{i}^{2}$ is
given by $s_{N}(u)=\frac{|\mathbb{S}^{N-1}|}{2}u^{\frac{N}{2}-1}Z_{N}(f,\sqrt{u})$. \end{lem}
\begin{proof}
Similar to Lemma \ref{lem: Definition of h} for any $\varphi=\varphi(r)$
we find that\[
\mathbb{E}\varphi=\int_{0}^{\infty}\varphi(r)\left(\int_{\mathbb{S}^{N-1}(r)}f(v_{1})\dots f(v_{N})ds_{r}^{N}\right)dr=\int_{0}^{\infty}\varphi(r)|\mathbb{S}^{N-1}|r^{N-1}Z_{N}(f,r)dr\]
 on the other hand\[
\mathbb{E}\varphi=\int_{0}^{\infty}\varphi(\sqrt{x})s_{N}(x)dx=\int_{0}^{\infty}\varphi(r)\cdot2r\cdot s_{N}\left(r^{2}\right)dr\]
 Since $\varphi$ is arbitrary\[
2r\cdot s_{N}\left(r^{2}\right)=|\mathbb{S}^{N-1}|r^{N-1}Z_{N}(f,r)\]
which implies the result.\end{proof}
\begin{cor}
\label{cor:Expression-for Z}(Expression for $Z_{N}(f,r)$) Under
the conditions of Lemma \ref{lem:Definition of s}\[
Z_{N}(f,\sqrt{r})=\frac{2h^{*N}(r)}{|\mathbb{S}^{N-1}|r^{\frac{N}{2}-1}}\]
 where $h^{^{*N}}$ is the $N$-fold convolution of $h$, defined
in Lemma \ref{lem: Definition of h}.\end{cor}
\begin{proof}
This follows immediately from Lemma \ref{lem:Definition of s}, Lemma
\ref{lem: Definition of h} and a known probability fact. 
\end{proof}

\section{Central Limit Theorem\label{sec:Centrel-Limit-Theorem}}

In order for us to be able to compute the entropy production an asymptotic
behavior for $Z_{N}(f,r)$ is needed. As seen in Section \ref{sec:About-the-Function Z}
the function $Z_{N}(f,r)$ is closely related to the $N$-fold convolution
of the density function $h(u)$ and as such we'll employ standard
techniques to estimate it. The specific function we'll construct as
a test function for the entropy production has the property that the
Fourier transform of its one particle function splits the line into
two natural domains: One where we can use analytic expansion, and
one where the decay is dominated by exponential functions. The radius
of the separating circle would depend on a parameter $\delta=$$\delta_{N}$
that we'll exploit later on to get the final conclusion.

While this is the case arising in our specific construction, we believe
that it's a natural way to view the problem. Even though we have yet
to attempt any different test functions we think that similar situation
would happen in a larger class of functions created from one particle
function. As such, a generalization of our computation was made and
is presented in this section.

The reader should keep in mind the following intuition while reading
this section: $g(\xi)$ represents the Fourier transform of the function
$h(u)$, connected to the one particle function via Lemma \ref{lem: Definition of h}.
The first lemma of the section explores the domain outside the radius
of analiticity while the second explores the domain where analytic
expansion is possible. Lastly, the parameter$\delta$ is a function
of $N$, going to zero as $N$ goes to infinity. 
\begin{lem}
\label{lem:Non-Analycity Domain}Let $g_{\delta}(\xi)=g_{\delta_{N}}(\xi)$
be such that

$(i)$ for $|\xi|>c\delta$ $|g_{\delta}(\xi)|\leq1-\alpha(\delta)$,
where $\alpha(\delta)>0$.

$(ii)$ $|g_{\delta}(\xi)|\leq1$ for all $\xi$.

Then\[
\int_{|\xi|>c\delta}\left|g_{\delta}^{N}(\xi)-\gamma_{1}^{N}(\xi)\right|d\xi\]
 \[
\leq2\int_{|\xi|>c\delta}\left|g_{\delta}(\xi)\right|^{N-1}d\xi+\frac{\left(1-\alpha(\delta)\right)^{\frac{N}{2}-1}}{\pi c\delta\Sigma_{\delta}^{2}}+\frac{1}{\pi c\delta\Sigma_{\delta}^{2}}\cdot e^{-(1+N)\pi^{2}c^{2}\delta^{2}\Sigma_{\delta}^{2}}\]
 where $\gamma_{1}(\xi)=e^{-2\pi i\zeta}\cdot e^{-2\pi^{2}\xi^{2}\Sigma_{\delta}^{2}}$.\end{lem}
\begin{proof}
We have that\[
\int_{|\xi|>c\delta}\left|g_{\delta}^{N}(\xi)-\gamma_{1}^{N}(\xi)\right|d\xi=\int_{|\xi|>c\delta}\left|g_{\delta}(\xi)-\gamma_{1}(\xi)\right|\cdot\left|\sum_{k=0}^{N-1}g_{\delta}^{N-k-1}(\xi)\gamma_{1}^{k}(\xi)\right|d\xi\]
 \[
\leq2\int_{|\xi|>c\delta}\sum_{k=0}^{N-1}\left|g_{\delta}^{N-k-1}(\xi)\right|\left|\gamma_{1}^{k}(\xi)\right|d\xi\]
 \[
\leq2\int_{|\xi|>c\delta}\left|g_{\delta}(\xi)\right|^{N-1}d\xi+2\sum_{k=1}^{N-1}\left(1-\alpha(\delta)\right)^{N-k-1}\int_{|\xi|>c\delta}e^{-2k\pi^{2}\xi^{2}\Sigma_{\delta}^{2}}d\xi\]
 Using Lemma \ref{lem:Gaussian-integral-estimation} and \ref{lem:Special-Sums-Evaluation}
in the Appendix we find that\[
\sum_{k=k_{0}}^{N-1}\int_{|\xi|>c\delta}e^{-2k\pi^{2}\xi^{2}\Sigma_{\delta}^{2}}d\xi\leq\sum_{k=k_{0}}^{N-1}\frac{\sqrt{2\pi}\cdot e^{-\frac{4k\pi^{2}c^{2}\delta^{2}\Sigma_{\delta}^{2}}{2}}}{\sqrt{4k\pi^{2}\Sigma_{\delta}^{2}}}\]
 \[
\leq\frac{1}{2\pi c\delta\Sigma_{\delta}^{2}}\cdot e^{-2k_{0}\pi^{2}c^{2}\delta^{2}\Sigma_{\delta}^{2}}\]
 Hence\[
\int_{|\xi|>c\delta}\left|g_{\delta}^{N}(\xi)-\gamma_{1}^{N}(\xi)\right|d\xi\]
 \[
\leq2\int_{|\xi|>c\delta}\left|g_{\delta}(\xi)\right|^{N-1}d\xi+2\left(1-\alpha(\delta)\right)^{N-\left[\frac{N}{2}\right]-1}\sum_{k=1}^{\left[\frac{N}{2}\right]}\int_{|\xi|>c\delta}e^{-2k\pi^{2}\xi^{2}\Sigma_{\delta}^{2}}d\xi\]
 \[
+2\sum_{k=\left[\frac{N}{2}\right]+1}^{N-1}\int_{|\xi|>c\delta}e^{-2k\pi^{2}\xi^{2}\Sigma_{\delta}^{2}}d\xi\]
 \[
\leq2\int_{|\xi|>c\delta}\left|g_{\delta}(\xi)\right|^{N-1}d\xi+\frac{\left(1-\alpha(\delta)\right)^{\frac{N}{2}-1}}{\pi c\delta\Sigma_{\delta}^{2}}+\frac{1}{\pi c\delta\Sigma_{\delta}^{2}}\cdot e^{-(1+N)\pi^{2}c^{2}\delta^{2}\Sigma_{\delta}^{2}}\]
 \end{proof}
\begin{lem}
\label{lem:Analytic Domain}Let $g_{\delta}(\xi)=g_{\delta_{N}}(\xi)$
be such that

$(i)$ there exist $M_{0},M_{1},M_{2}>0$ such that $\sup_{|\xi|<c\delta}\left|g_{\delta}(\xi)-\gamma_{1}(\xi)\right|\leq\left(\frac{M_{0}}{\delta^{2}}+\frac{M_{1}}{\delta}+M_{2}\right)|\xi|^{3}$.

$(ii)$ for $c\delta^{1+\beta}<|\xi|<c\delta$ $|g_{\delta}(\xi)|\leq1-\alpha_{\beta}(\delta)$
where $\alpha_{\beta}(\delta)>0$.

$(iii)$ $|g_{\delta}(\xi)|\leq1$ for all $\xi$.

Then \[
\int_{|\xi|<c\delta}\left|g_{\delta}^{N}(\xi)-\gamma_{1}^{N}(\xi)\right|d\xi\leq\frac{c^{4}\delta^{2}\left(M_{0}+M_{1}\delta+M_{2}\delta^{2}\right)}{2}\]
 \[
+\frac{c^{3}\delta\sqrt{N}\left(M_{0}+M_{1}\delta+M_{2}\delta^{2}\right)\left(1-\alpha_{\beta}(\delta)\right)^{\frac{N}{2}-1}}{\sqrt{\pi\Sigma_{\delta}^{2}}}+\frac{c^{3}\delta^{1-\beta}\left(M_{0}+M_{1}\delta+M_{2}\delta^{2}\right)e^{-\pi^{2}(N-1)c^{2}\delta^{2+2\beta}\Sigma_{\delta}^{2}}}{2\pi c\delta\Sigma_{\delta}^{2}\cdot\sqrt{1-e^{-2\pi^{2}Nc^{2}\delta^{2}\Sigma_{\delta}^{2}}}}\]
 \[
+\frac{2c^{3}\left(M_{0}+M_{1}\delta+M_{2}\delta^{2}\right)\sqrt{N}\delta^{1+3\beta}}{\sqrt{2\pi\Sigma_{\delta}^{2}}}\]
 where $\gamma_{1}(\xi)=e^{-2\pi i\zeta}\cdot e^{-2\pi^{2}\xi^{2}\Sigma_{\delta}^{2}}$.\end{lem}
\begin{rem}
The coefficients $M_{0},M_{1}$and $M_{2}$ play a major role in the
estimation. Notice that we can get a better result if have that $M_{0}=0$
and an even better result if both $M_{0}$ and $M_{1}$ are zero. \end{rem}
\begin{proof}
Similar to Lemma \ref{lem:Non-Analycity Domain} we find that\[
\int_{|\xi|<c\delta}\left|g_{\delta}^{N}(\xi)-\gamma_{1}^{N}(\xi)\right|d\xi\leq\sum_{k=0}^{N-1}\int_{|\xi|<c\delta}\left|g_{\delta}(\xi)-\gamma_{1}(\xi)\right|\left|g_{\delta}(\xi)\right|^{N-k-1}\left|\gamma_{1}(\xi)\right|^{k}d\xi\]
 \[
\leq\int_{|\xi|<c\delta}\left(\frac{M_{0}}{\delta^{2}}+\frac{M_{1}}{\delta}+M_{2}\right)|\xi|^{3}d\xi+\sum_{k=1}^{N-1}\int_{|\xi|<c\delta}\left(\frac{M_{0}}{\delta^{2}}+\frac{M_{1}}{\delta}+M_{2}\right)|\xi|^{3}\left|g_{\delta}(\xi)\right|^{N-k-1}\left|\gamma_{1}(\xi)\right|^{k}d\xi\]
 \[
=\frac{c^{4}\delta^{2}\left(M_{0}+M_{1}\delta+M_{2}\delta^{2}\right)}{2}+\sum_{k=1}^{N-1}\int_{c\delta^{1+\beta}<|\xi|<c\delta}\left(\frac{M_{0}}{\delta^{2}}+\frac{M_{1}}{\delta}+M_{2}\right)|\xi|^{3}\left|g_{\delta}(\xi)\right|^{N-k-1}\left|\gamma_{1}(\xi)\right|^{k}d\xi\]
 \[
+\sum_{k=1}^{N-1}\int_{|\xi|<c\delta^{1+\beta}}\left(\frac{M_{0}}{\delta^{2}}+\frac{M_{1}}{\delta}+M_{2}\right)|\xi|^{3}\left|g_{\delta}(\xi)\right|^{N-k-1}\left|\gamma_{1}(\xi)\right|^{k}d\xi\]
 We have that\[
\sum_{k=1}^{N-1}\int_{c\delta^{1+\beta}<|\xi|<c\delta}\left(\frac{M_{0}}{\delta^{2}}+\frac{M_{1}}{\delta}+M_{2}\right)|\xi|^{3}\left|g_{\delta}(\xi)\right|^{N-k-1}\left|\gamma_{1}(\xi)\right|^{k}d\xi\]
 \[
\leq c^{3}\delta\left(M_{0}+M_{1}\delta+M_{2}\delta^{2}\right)\sum_{k=1}^{N-1}\left(1-\alpha_{\beta}(\delta)\right)^{N-k-1}\int_{c\delta^{1+\beta}<|\xi|<c\delta}e^{-2k\pi^{2}\xi^{2}\Sigma_{\delta}^{2}}d\xi\]
 \[
\leq c^{3}\delta\left(M_{0}+M_{1}\delta+M_{2}\delta^{2}\right)\left(1-\alpha_{\beta}(\delta)\right)^{\frac{N}{2}-1}\sum_{k=1}^{\left[\frac{N}{2}\right]}\int_{|\xi|<c\delta}e^{-2k\pi^{2}\xi^{2}\Sigma_{\delta}^{2}}d\xi\]
 \[
+c^{3}\delta\left(M_{0}+M_{1}\delta+M_{2}\delta^{2}\right)\sum_{k=\left[\frac{N}{2}\right]+1}^{N-1}\int_{c\delta^{1+\beta}<|\xi|<c\delta}e^{-2k\pi^{2}\xi^{2}\Sigma_{\delta}^{2}}d\xi\]
 \[
\leq c^{3}\delta\left(M_{0}+M_{1}\delta+M_{2}\delta^{2}\right)\left(1-\alpha_{\beta}(\delta)\right)^{\frac{N}{2}-1}\sum_{k=1}^{\left[\frac{N}{2}\right]}\frac{\sqrt{1-e^{-4\pi^{2}kc^{2}\delta^{2}\Sigma_{\delta}^{2}}}}{\sqrt{2\pi\Sigma_{\delta}^{2}k}}\]
 \[
+c^{3}\delta\left(M_{0}+M_{1}\delta+M_{2}\delta^{2}\right)\sum_{k=\left[\frac{N}{2}\right]+1}^{N-1}\left(\int_{|\xi|<c\delta}e^{-2k\pi^{2}\xi^{2}\Sigma_{\delta}^{2}}d\xi-\int_{|\xi|<c\delta^{1+\beta}}e^{-2k\pi^{2}\xi^{2}\Sigma_{\delta}^{2}}d\xi\right)\]
 \[
\leq c^{3}\delta\left(M_{0}+M_{1}\delta+M_{2}\delta^{2}\right)\left(1-\alpha_{\beta}(\delta)\right)^{\frac{N}{2}-1}\sum_{k=1}^{\left[\frac{N}{2}\right]}\frac{1}{\sqrt{2\pi\Sigma_{\delta}^{2}k}}\]
 \[
+c^{3}\delta\left(M_{0}+M_{1}\delta+M_{2}\delta^{2}\right)\sum_{k=\left[\frac{N}{2}\right]+1}^{N-1}\frac{\left(\sqrt{1-e^{-4\pi^{2}kc^{2}\delta^{2}\Sigma_{\delta}^{2}}}-\sqrt{1-e^{-2\pi^{2}kc^{2}\delta^{2+2\beta}\Sigma_{\delta}^{2}}}\right)}{\sqrt{2\pi k\Sigma_{\delta}^{2}}}\]
 \[
\leq\frac{c^{3}\delta\left(M_{0}+M_{1}\delta+M_{2}\delta^{2}\right)\left(1-\alpha_{\beta}(\delta)\right)^{\frac{N}{2}-1}}{\sqrt{2\pi\Sigma_{\delta}^{2}}}\cdot\sqrt{4\left[\frac{N}{2}\right]}\]
 \[
+\frac{c^{3}\delta\left(M_{0}+M_{1}\delta+M_{2}\delta^{2}\right)}{\sqrt{2\pi\Sigma_{\delta}^{2}}}\sum_{k=\left[\frac{N}{2}\right]+1}^{N-1}\frac{1}{\sqrt{k}}\cdot\frac{e^{-2\pi^{2}kc^{2}\delta^{2+2\beta}\Sigma_{\delta}^{2}}-e^{-4\pi^{2}kc^{2}\delta^{2}\Sigma_{\delta}^{2}}}{\left(\sqrt{1-e^{-4\pi^{2}kc^{2}\delta^{2}\Sigma_{\delta}^{2}}}+\sqrt{1-e^{-2\pi^{2}kc^{2}\delta^{2+2\beta}\Sigma_{\delta}^{2}}}\right)}\]
 \[
\leq\frac{c^{3}\delta\sqrt{N}\left(M_{0}+M_{1}\delta+M_{2}\delta^{2}\right)\left(1-\alpha_{\beta}(\delta)\right)^{\frac{N}{2}-1}}{\sqrt{\pi\Sigma_{\delta}^{2}}}\]
 \[
+\frac{c^{3}\delta\left(M_{0}+M_{1}\delta+M_{2}\delta^{2}\right)}{\sqrt{2\pi\Sigma_{\delta}^{2}}}\sum_{k=\left[\frac{N}{2}\right]+1}^{N-1}\frac{1}{\sqrt{k}}\cdot\frac{e^{-2\pi^{2}kc^{2}\delta^{2+2\beta}\Sigma_{\delta}^{2}}}{\sqrt{1-e^{-4\pi^{2}kc^{2}\delta^{2}\Sigma_{\delta}^{2}}}}\]
 \[
\leq\frac{c^{3}\delta\sqrt{N}\left(M_{0}+M_{1}\delta+M_{2}\delta^{2}\right)\left(1-\alpha_{\beta}(\delta)\right)^{\frac{N}{2}-1}}{\sqrt{\pi\Sigma_{\delta}^{2}}}\]
 \[
+\frac{c^{3}\delta\left(M_{0}+M_{1}\delta+M_{2}\delta^{2}\right)}{\sqrt{2\pi\Sigma_{\delta}^{2}}\cdot\sqrt{1-e^{-2\pi^{2}Nc^{2}\delta^{2}\Sigma_{\delta}^{2}}}}\sum_{k=\left[\frac{N}{2}\right]+1}^{N-1}\frac{e^{-2\pi^{2}kc^{2}\delta^{2+2\beta}\Sigma_{\delta}^{2}}}{\sqrt{k}}\]
 \[
\leq\frac{c^{3}\delta\sqrt{N}\left(M_{0}+M_{1}\delta+M_{2}\delta^{2}\right)\left(1-\alpha_{\beta}(\delta)\right)^{\frac{N}{2}-1}}{\sqrt{\pi\Sigma_{\delta}^{2}}}+\frac{c^{3}\delta^{1-\beta}\left(M_{0}+M_{1}\delta+M_{2}\delta^{2}\right)e^{-\pi^{2}(N-1)c^{2}\delta^{2+2\beta}\Sigma_{\delta}^{2}}}{2\pi c\delta\Sigma_{\delta}^{2}\cdot\sqrt{1-e^{-2\pi^{2}Nc^{2}\delta^{2}\Sigma_{\delta}^{2}}}}\]
 Next we find that\[
\sum_{k=1}^{N-1}\int_{|\xi|<c\delta^{1+\beta}}\left(\frac{M_{0}}{\delta^{2}}+\frac{M_{1}}{\delta}+M_{2}\right)|\xi|^{3}\left|g_{\delta}(\xi)\right|^{N-k-1}\left|\gamma_{1}(\xi)\right|^{k}d\xi\]
 \[
\leq c^{3}\left(M_{0}+M_{1}\delta+M_{2}\delta^{2}\right)\delta^{1+3\beta}\cdot\sum_{k=1}^{N-1}\int_{|\xi|<c\delta^{1+\beta}}e^{-2k\pi^{2}\xi^{2}\Sigma_{\delta}^{2}}d\xi\]
 \[
\leq c^{3}\left(M_{0}+M_{1}\delta+M_{2}\delta^{2}\right)\delta^{1+3\beta}\cdot\sum_{k=1}^{N-1}\frac{\sqrt{1-e^{-4k\pi^{2}c^{2}\delta^{2+2\beta}\Sigma_{\delta}^{2}}}}{\sqrt{2\pi k\Sigma_{\delta}^{2}}}\]
 \[
\leq\frac{c^{3}\left(M_{0}+M_{1}\delta+M_{2}\delta^{2}\right)\delta^{1+3\beta}}{\sqrt{2\pi\Sigma_{\delta}^{2}}}\cdot\sum_{k=1}^{N-1}\frac{1}{\sqrt{k}}\]
 \[
\leq\frac{2c^{3}\left(M_{0}+M_{1}\delta+M_{2}\delta^{2}\right)\sqrt{N}\delta^{1+3\beta}}{\sqrt{2\pi\Sigma_{\delta}^{2}}}\]
 Which completes the proof.\end{proof}
\begin{thm}
\label{thm:uniform approximation of convulotion}Let $h_{\delta}(x)=h_{\delta_{N}}(x)$
be a function such that $g_{\delta}(\xi)=\widehat{h_{\delta}}(\xi)$
satisfies

$(i)$ for $|\xi|>c\delta_{N}$ $|g_{\delta_{N}}(\xi)|\leq1-\alpha(\delta_{N})$,
where $\alpha(\delta_{N})>0$

$(ii)$ there exist $M_{0},M_{1},M_{2}>0$ such that $\sup_{|\xi|<c\delta_{N}}\left|g_{\delta_{N}}(\xi)-\gamma_{1}(\xi)\right|\leq\left(\frac{M_{0}}{\delta_{N}^{2}}+\frac{M_{1}}{\delta_{N}}+M_{2}\right)|\xi|^{3}$

$(iii)$ for $c\delta_{N}^{1+\beta}<|\xi|<c\delta_{N}$ $|g_{\delta_{N}}(\xi)|\leq1-\alpha_{\beta}(\delta_{N})$
where $\alpha_{\beta}(\delta_{N})>0$

$(vi)$ $|g_{\delta_{N}}(\xi)|\leq1$ for all $\xi$

and if\begin{equation}
\begin{array}{c}
\delta_{N},\alpha(\delta_{N})\, and\,\alpha_{\beta}(\delta_{N})\, are\, domianted\, by\, powers\, of\, N\\
\alpha(\delta_{N})N\underset{N\rightarrow\infty}{\longrightarrow}\infty\\
\alpha_{\beta}(\delta_{N})N\underset{N\rightarrow\infty}{\longrightarrow}\infty\\
\Sigma_{\delta_{N}}^{2}\delta_{N}^{2+2\beta}N\underset{N\rightarrow\infty}{\longrightarrow}\infty\\
\delta_{N}^{1+3\beta}N\underset{N\rightarrow\infty}{\longrightarrow}0\\
\sqrt{N}\Sigma_{\delta_{N}}\int_{|\xi|>c\delta_{N}}\left|g_{\delta_{N}}(\xi)\right|^{N-1}d\xi\underset{N\rightarrow\infty}{\longrightarrow}0\\
\delta_{N}^{\frac{3}{2}(1-\beta)}\Sigma_{\delta_{N}}\, is\, bounded\end{array}\label{eq:condition}\end{equation}
 then\[
\sup_{x}\left|h_{\delta_{N}}^{*N}(x)-\frac{1}{\sqrt{N}\Sigma_{\delta_{N}}}\cdot\frac{e^{-\frac{\left(x-N\right)^{2}}{2N\Sigma_{\delta_{N}}^{2}}}}{\sqrt{2\pi}}\right|\leq\frac{\epsilon(N)}{\sqrt{N}\Sigma_{\delta_{N}}}\]
 where $h_{\delta_{N}}^{*N}(x)$ is the $N$-fold convolution and
$\epsilon(N)\underset{N\rightarrow\infty}{\longrightarrow}0$. \end{thm}
\begin{proof}
It is easy to check that $\widehat{\frac{1}{\sqrt{N}\Sigma_{\delta}}\cdot\frac{e^{-\frac{\left(x-N\right)^{2}}{2N\Sigma_{\delta}^{2}}}}{\sqrt{2\pi}}}(\xi)=\gamma_{1}^{N}(\xi)$

Using Lemma \ref{lem:Non-Analycity Domain} and \ref{lem:Analytic Domain}
we find that \[
\sup_{x}\left|h_{\delta}^{*N}(x)-\frac{1}{\sqrt{N}\Sigma_{\delta}}\cdot\frac{e^{-\frac{\left(x-N\right)^{2}}{2N\Sigma_{\delta}^{2}}}}{\sqrt{2\pi}}\right|\leq\int_{\mathbb{R}}\left|g_{\delta}^{N}(\xi)-\gamma_{1}^{N}(\xi)\right|d\xi\]
 \[
=\int_{|\xi|<c\delta}\left|g_{\delta}^{N}(\xi)-\gamma_{1}^{N}(\xi)\right|d\xi+\int_{|\xi|>c\delta}\left|g_{\delta}^{N}(\xi)-\gamma_{1}^{N}(\xi)\right|d\xi\]
 \[
\leq\frac{1}{\sqrt{N}\Sigma_{\delta}}\left(\frac{c^{4}\sqrt{N\delta^{1+3\beta}}\delta^{\frac{3}{2}(1-\beta)}\Sigma_{\delta}\left(M_{0}+M_{1}\delta+M_{2}\delta^{2}\right)}{2}\right.\]
 \[
+\frac{c^{3}\delta N\left(M_{0}+M_{1}\delta+M_{2}\delta^{2}\right)\left(1-\alpha_{\beta}(\delta)\right)^{\frac{N}{2}-1}}{\sqrt{\pi}}+\frac{c^{3}\sqrt{N}\delta^{1-\beta}\left(M_{0}+M_{1}\delta+M_{2}\delta^{2}\right)e^{-\pi^{2}(N-1)c^{2}\delta^{2+2\beta}\Sigma_{\delta}^{2}}}{2\pi c\delta\Sigma_{\delta}\cdot\sqrt{1-e^{-2\pi^{2}Nc^{2}\delta^{2}\Sigma_{\delta}^{2}}}}\]
 \[
+\frac{2c^{3}\left(M_{0}+M_{1}\delta+M_{2}\delta^{2}\right)N\delta^{1+3\beta}}{\sqrt{2\pi}}+2\sqrt{N}\Sigma_{\delta}\int_{|\xi|>c\delta}\left|g_{\delta}(\xi)\right|^{N-1}d\xi\]
 \[
\left.+2\left(1-\alpha(\delta)\right)^{\frac{N}{2}-1}\cdot\frac{\sqrt{N}}{2\pi c\delta\Sigma_{\delta}}+\frac{\sqrt{N}}{\pi c\delta\Sigma_{\delta}}\cdot e^{-(1+N)\pi^{2}c^{2}\delta^{2}\Sigma_{\delta}^{2}}\right)\]
 Conditions (\ref{eq:condition}) insure the desired conclusion.\end{proof}
\begin{rem}
\label{rem:epsilon j}A careful look at the proof of Theorem \ref{thm:uniform approximation of convulotion}
shows that for a fixed $j$ if $\lim_{N\rightarrow\infty}\sqrt{N-j}\Sigma_{\delta_{N}}\int_{|\xi|>c\delta_{N}}\left|g_{\delta_{N}}(\xi)\right|^{N-j-1}d\xi=0$
and conditions (\ref{eq:condition}) are satisfied (with the obvious
change) then\[
\sup_{x}\left|h_{\delta_{N}}^{*N-j}(x)-\frac{1}{\sqrt{N-j}\Sigma_{\delta_{N}}}\cdot\frac{e^{-\frac{\left(x-N+j\right)^{2}}{2(N-j)\Sigma_{\delta_{N}}^{2}}}}{\sqrt{2\pi}}\right|\leq\frac{\epsilon_{j}(N)}{\sqrt{N-j}\Sigma_{\delta_{N}}}\]
 where $\epsilon_{j}(N)\underset{N\rightarrow\infty}{\longrightarrow}0$.
\end{rem}

\section{Entropy Production and Villani's Conjecture\label{sec:Entropy-Production}}

In this section we'll find an exact estimation for the entropy production.
The idea behind this estimation is to use superposition of stationary
solutions for the Boltzmann equation: the Maxwellian densities $M_{a}(v)=\frac{e^{-\frac{b^{2}}{2a}}}{\sqrt{2\pi a}}$.
This idea was exploited by Carlen, Carvalho, Le Roux, Loss, and Villani
(\cite{key-7}) and Bobylev and Cercignani (\cite{key-1}) before
them.

The basic one particle function would be\[
f_{\delta_{N}}(v)=f_{\delta}(v)=\delta M_{\frac{1}{2\delta}}(v)+(1-\delta)M_{\frac{1}{2(1-\delta)}}(v)\]
 This function has the property that both its parts have the same
energy\[
\int_{\mathbb{R}}\delta M_{\frac{1}{2\delta}}(v)dv=\int_{\mathbb{R}}(1-\delta)M_{\frac{1}{2(1-\delta)}}(v)dv=\frac{1}{2}\]
 while as $\delta$ gets smaller the number of particles represented
by $\delta M_{\frac{1}{2\delta}}(v)$ is far smaller than those represented
by $(1-\delta)M_{\frac{1}{2(1-\delta)}}(v)$. The fact that we have
a small number of very energetic particles and a large number of very
stable particles trying to equilibrate will cause slow decay into
equilibrium. That physical intuition is indeed true as would be seen
shortly. 
\begin{lem}
\label{lem:Properties of h}Let $h_{\delta}(u)=\frac{f_{\delta}(\sqrt{u})+f_{\delta}(-\sqrt{u})}{2\sqrt{u}}=\frac{f_{\delta}(\sqrt{u})}{\sqrt{u}}$
then

$(i)$ $\int_{0}^{\infty}h_{\delta}(u)du=1$

$(ii)$ $\int_{0}^{\infty}uh_{\delta}(u)du=1$

$(iii)$ $\Sigma_{\delta}^{2}=\int_{0}^{\infty}u^{2}h_{\delta}(u)du-\left(\int_{0}^{\infty}uh_{\delta}(u)du\right)^{2}=\frac{3}{4\delta(1-\delta)}-1$

$(iv)$ $\widehat{h_{\delta}}(\xi)=\frac{\delta}{\sqrt{1+\frac{2\pi i\xi}{\delta}}}+\frac{1-\delta}{\sqrt{1+\frac{2\pi i\xi}{1-\delta}}}$\end{lem}
\begin{proof}
$(i)-(iii)$ follow immediately from the fact that $\int_{0}^{\infty}u^{m}h_{\delta}(u)du=\int_{\mathbb{R}}x^{2m}f_{\delta}(x)dx$
and the fact that\[
\int_{\mathbb{R}}M_{a}(u)du=1,\,\,\int_{\mathbb{R}}u^{2}M_{a}(u)du=a,\,\,\int_{\mathbb{R}}u^{4}M_{a}(u)du=3a^{2}\]
 We're only left with proving $(iv)$.

It is easy to check that\[
\frac{d}{d\xi}\int_{\mathbb{R}}M_{a}(u)\cdot e^{-2\pi i\xi u^{2}}du=\frac{-2\pi ia}{1+4\pi ia\xi}\int_{\mathbb{R}}M_{a}(u)\cdot e^{-2\pi i\xi u^{2}}du\]
 The initial value problem $\frac{d}{d\xi}\varphi(\xi)=\frac{-2\pi ia}{1+4\pi ia\xi}\varphi(\xi),\,\,\xi\in\mathbb{R}$,
$\varphi(0)=1$ has the unique solution\[
\varphi(\xi)=\frac{1}{\sqrt{1+4\pi ia\xi}}\]
 Thus, the result follows from the definition of $f_{\delta}$ and
the fact that \[
\widehat{h_{\delta}}(\xi)=\int_{0}^{\infty}h_{\delta}(u)e^{-2\pi i\xi u}du=\int_{\mathbb{R}}f_{\delta}(u)e^{-2\pi i\xi u^{2}}du\]
 \end{proof}
\begin{lem}
\label{lem:Special properties of h}Let $g_{\delta}(\xi)=\widehat{h_{\delta}}(\xi)$
where $\delta<\frac{1}{2}$ then

$(i)$ for $|\xi|>\frac{\delta}{4\pi}$ $|g_{\delta}(\xi)|\leq1-\delta\left(1-\sqrt[4]{\frac{4}{5}}\right)+\rho_{1}(\delta)$
where $\frac{\rho_{1}(\delta)}{\delta}\underset{\delta\rightarrow0}{\longrightarrow}0$

$(ii)$ there exist $M_{0},M_{1},M_{2}>0$ such that $\sup_{|\xi|<\frac{\delta}{4\pi}}\left|g_{\delta}(\xi)-\gamma_{1}(\xi)\right|\leq\left(\frac{M_{0}}{\delta^{2}}+\frac{M_{1}}{\delta}+M_{2}\right)|\xi|^{3}$.

$(iii)$ for $\frac{\delta^{1+\beta}}{4\pi}<|\xi|<\frac{\delta}{4\pi}$
$|g_{\delta}(\xi)|\leq1-\frac{\delta^{1+2\beta}}{16}+\rho_{2}(\delta)$
where $\frac{\rho_{2}(\delta)}{\delta^{1+2\beta}}\underset{\delta\rightarrow0}{\longrightarrow}0$

$(vi)$ $|g_{\delta}(\xi)|\leq1$ for all $\xi$.

$(v)$ for a fixed $j$ $\int_{|\xi|>\frac{\delta}{4\pi}}\left|g_{\delta_{N}}(\xi)\right|^{N-j-1}d\xi\leq\frac{\left(1-\delta\left(1-\sqrt[4]{\frac{4}{5}}\right)+\rho_{1}(\delta)\right)^{N-j-1}}{\pi}+\frac{2}{\pi(N-j)}$\end{lem}
\begin{proof}
$(i)$ For $|\xi|>\frac{\delta}{4\pi}$ \[
\left|g_{\delta}(\xi)\right|\leq\frac{\delta}{\sqrt[4]{1+\frac{4\pi^{2}\xi^{2}}{\delta^{2}}}}+\frac{1-\delta}{\sqrt[4]{1+\frac{4\pi^{2}\xi^{2}}{(1-\delta)^{2}}}}\leq\frac{\delta}{\sqrt[4]{\frac{5}{4}}}+\frac{1-\delta}{\sqrt[4]{1+\frac{\delta^{2}}{4(1-\delta)^{2}}}}\]
 \[
=\sqrt[4]{\frac{4}{5}}\delta+(1-\delta)\left(1-\frac{\delta^{2}}{16(1-\delta)^{2}}+\dots\right)=1-\delta\left(1-\sqrt[4]{\frac{4}{5}}\right)+\rho_{1}(\delta)\]
 where $\frac{\rho_{1}(\delta)}{\delta}\underset{\delta\rightarrow0}{\longrightarrow}0$.

$(ii)$ Using the expansions for $\frac{1}{\sqrt{1+x}}$ and $e^{x}$
we find that for $|\xi|<\frac{\delta}{4\pi}$\[
\left|h_{\delta}(\xi)-\gamma_{1}(\xi)\right|\leq|\xi|^{3}\left(\frac{8\pi^{3}}{\delta^{2}}\cdot\left|\phi\left(\frac{2\pi i\xi}{\delta}\right)\right|+\frac{8\pi^{3}}{(1-\delta)^{2}}\cdot\left|\phi\left(\frac{2\pi i\xi}{1-\delta}\right)\right|\right.\]
 \[
+\frac{3\pi^{3}}{\delta(1-\delta)}-4\pi^{3}+2\pi^{4}\left(\frac{3}{4\delta(1-\delta)}-1\right)^{2}|\xi|+\frac{3\pi^{4}}{\delta(1-\delta)}|\xi|-4\pi^{4}|\xi|\]
 \[
+4\pi^{5}\left(\frac{3}{4\delta(1-\delta)}-1\right)^{2}|\xi|^{2}+4\pi^{6}\left(\frac{3}{4\delta(1-\delta)}-1\right)^{2}|\xi|^{3}\]
 \[
\left.+8\pi^{3}\left|\psi\left(-2\pi i\xi\right)\right|+8\pi^{6}\left(\frac{3}{4\delta(1-\delta)}-1\right)^{3}|\xi|^{3}\left|\psi\left(-2\pi^{2}\Sigma_{\delta}^{2}\xi^{2}\right)\right|\right)\]
 where $\phi(x)$ is analytic in $|x|<\frac{1}{2}$ and $\psi(x)$
is an entire function. Denoting $M_{\phi}=\sup_{|x|\leq\frac{1}{2}}\left|\phi(x)\right|$
and $M_{\psi}=\sup_{|x|\leq\frac{1}{2}}\left|\psi(x)\right|$ we find
that \[
\left|h_{\delta}(\xi)-\gamma_{1}(\xi)\right|\leq\left(\frac{8\pi^{3}}{\delta^{2}}M_{\phi}+\frac{57\pi^{3}}{8\delta}+\pi^{3}\left(32M_{\phi}+\frac{141}{64}+\frac{539}{64}M_{\psi}\right)\right)|\xi|^{3}\]

$(iii)$ For $|\xi|>\frac{\delta^{1+\beta}}{4\pi}$ \[
\left|g_{\delta}(\xi)\right|\leq\frac{\delta}{\sqrt[4]{1+\frac{4\pi^{2}\xi^{2}}{\delta^{2}}}}+\frac{1-\delta}{\sqrt[4]{1+\frac{4\pi^{2}\xi^{2}}{(1-\delta)^{2}}}}\leq\frac{\delta}{\sqrt[4]{1+\frac{\delta^{2\beta}}{4}}}+\frac{1-\delta}{\sqrt[4]{1+\frac{\delta^{2+2\beta}}{4(1-\delta)^{2}}}}\]
 \[
=\delta\left(1-\frac{\delta^{2\beta}}{16}+\dots\right)+\left(1-\delta\right)\left(1-\frac{\delta^{2+2\beta}}{16(1-\delta)^{2}}+\dots\right)=1-\frac{\delta^{1+2\beta}}{16}+\rho_{2}(\delta)\]
 where $\frac{\rho_{2}(\delta)}{\delta^{1+2\beta}}\underset{\delta\rightarrow0}{\longrightarrow}0$.

$(iv)$ This is a general property of the Fourier transform of a density
function.

$(v)$ \[
\int_{|\xi|>\frac{\delta}{4\pi}}\left|g_{\delta_{N}}(\xi)\right|^{N-j-1}d\xi\leq\int_{|\xi|>\frac{\delta}{4\pi}}\left(\frac{\delta}{\sqrt[4]{1+\frac{4\pi^{2}\xi^{2}}{\delta^{2}}}}+\frac{1-\delta}{\sqrt[4]{1+\frac{4\pi^{2}\xi^{2}}{(1-\delta)^{2}}}}\right)^{N-j-1}d\xi\]
 \[
=\frac{\delta}{2\pi}\int_{|x|>\frac{1}{2}}\left(\frac{\delta}{\sqrt[4]{1+x^{2}}}+\frac{1-\delta}{\sqrt[4]{1+\frac{\delta^{2}x^{2}}{(1-\delta)^{2}}}}\right)^{N-j-1}dx\]
 \[
\leq\frac{\left(1-\delta\left(1-\sqrt[4]{\frac{4}{5}}\right)+\rho_{1}(\delta)\right)^{N-j-1}}{\pi}+\frac{\delta}{\pi}\int_{\frac{1}{\delta}}^{\infty}\left(\frac{\delta^{\frac{3}{2}}}{\sqrt{\delta x}}+\frac{(1-\delta)^{\frac{3}{2}}}{\sqrt{\delta x}}\right)^{N-j-1}dx\]
 \[
\leq\frac{\left(1-\delta\left(1-\sqrt[4]{\frac{4}{5}}\right)+\rho_{1}(\delta)\right)^{N-j-1}}{\pi}+\frac{2}{\pi(N-j-3)}\]
 \end{proof}
\begin{rem}
\label{rem:eps j is ok}Note that in our case\[
\sqrt{N-j}\Sigma_{\delta_{N}}\int_{|\xi|>c\delta_{N}}\left|g_{\delta_{N}}(\xi)\right|^{N-j-1}d\xi\]
 \[
\leq\frac{3\sqrt{N-j}\left(1-\delta\left(1-\sqrt[4]{\frac{4}{5}}\right)+\rho_{1}(\delta)\right)^{N-j-1}}{2\pi\delta}+\frac{3}{\sqrt{N\delta}\cdot\sqrt{1-\frac{j+3}{N}}}\]
 so as long as the conditions in (\ref{eq:condition}) are satisfied
we have that $\epsilon_{j}(N)$ defined in Remark \ref{rem:epsilon j}
would satisfy $\epsilon_{j}(N)\underset{N\rightarrow\infty}{\longrightarrow}0$.\end{rem}
\begin{thm}
\label{thm:approximation of Z}Let $f_{\delta_{N}}(v)=f_{\delta}(v)=\delta M_{\frac{1}{2\delta}}(v)+(1-\delta)M_{\frac{1}{2(1-\delta)}}(v)$
such that \begin{equation}
\begin{array}{c}
\delta_{N}\, is\, domianted\, by\, powers\, of\, N\\
\begin{array}{c}
\delta_{N}^{1+2\beta}\cdot N\underset{N\rightarrow\infty}{\longrightarrow}\infty\\
\delta_{N}^{1+3\beta}\cdot N\underset{N\rightarrow\infty}{\longrightarrow}0\end{array}\end{array}\label{eq:delta conditions}\end{equation}
 then for a fixed $j$ \[
Z_{N-j}\left(f_{\delta_{N}},\sqrt{u}\right)=\frac{2}{\sqrt{N-j}\cdot\Sigma_{\delta_{N}}\cdot|\mathbb{S}^{N-j-1}|u^{\frac{N-j}{2}-1}}\left(\frac{e^{-\frac{\left(u-N+j\right)^{2}}{2(N-j)\Sigma_{\delta_{N}}^{2}}}}{\sqrt{2\pi}}+\lambda_{j}(N-j,u)\right)\]
 where $\sup_{u\in\mathbb{R}}\left|\lambda_{j}(N-j,u)\right|\leq\epsilon_{j}(N)$
and $\lim_{N\rightarrow\infty}\epsilon_{j}(N)=0$. \end{thm}
\begin{proof}
This is immediate from Lemma \ref{cor:Expression-for Z}, \ref{lem:Properties of h},
\ref{lem:Special properties of h}, Theorem \ref{thm:uniform approximation of convulotion}
and Remark \ref{rem:eps j is ok}. 
\end{proof}
We're now ready to compute the entropy production. We'll start by
estimating its denominator and numerator.
\begin{lem}
\label{lem:denominator of entropy production}Let $F_{N}\left(v_{1},\dots,v_{N}\right)=\frac{\Pi_{i=1}^{N}f_{\delta_{N}}(v_{i})}{Z_{N}(f,\sqrt{N})}$where
$\delta_{N}$ satisfies conditions (\ref{eq:delta conditions}). Then
\[
\lim_{N\rightarrow\infty}\frac{\int_{\mathbb{S}^{N-1}(\sqrt{N})}F_{N}\log F_{N}d\sigma^{N}}{N}=\frac{\log2}{2}\]
 \end{lem}
\begin{proof}
Using the symmetry of the problem, Lemma \ref{lem:Integration-on-the-Sphere II}
from the Appendix, Theorem \ref{thm:approximation of Z} and Stirling's
formula we find that\[
\int_{\mathbb{S}^{N-1}(\sqrt{N})}F_{N}\log F_{N}d\sigma^{N}=\frac{1}{Z_{N}(f_{\delta},\sqrt{N})}\cdot\sum_{k=1}^{N}\int_{\mathbb{S}^{N-1}(\sqrt{N})}\left(\Pi_{i=1}^{N}f_{\delta}(v_{i})\right)\log f_{\delta}(v_{k})d\sigma^{N}-\log Z_{N}(f_{\delta},\sqrt{N})\]
 \[
=\frac{N|\mathbb{S}^{N-2}|}{N^{\frac{N-2}{2}}|\mathbb{S}^{N-1}|}\int_{-\sqrt{N}}^{\sqrt{N}}f_{\delta}(v_{1})\log f_{\delta}(v_{1})\left(N-v_{1}^{2}\right)^{\frac{N-3}{2}}\cdot\frac{Z_{N-1}\left(f_{\delta},\sqrt{N-v_{1}^{2}}\right)}{Z_{N}(f_{\delta},\sqrt{N})}dv_{1}-\log Z_{N}(f_{\delta},\sqrt{N})\]
 \[
=\frac{N}{\sqrt{1-\frac{1}{N}}\left(1+\sqrt{2\pi}\lambda_{0}\left(N,N\right)\right)}\int_{\mathbb{R}}f_{\delta}(v_{1})\log f_{\delta}(v_{1})\cdot\chi_{[-\sqrt{N},\sqrt{N}]}(v_{1})\]
\[
\cdot\left(e^{-\frac{\left(1-v_{1}^{2}\right)^{2}}{(N-1)\Sigma_{\delta}^{2}}}+\sqrt{2\pi}\lambda_{1}\left(N-1,N-v_{1}^{2}\right)\right)dv_{1}\]
 \[
-\left(\log\left(\sqrt{2}\left(1+O\left(\frac{1}{\sqrt{N}}\right)\right)\left(1+\sqrt{2\pi}\lambda_{0}(N,N)\right)\right)-\frac{N}{2}\left(\log2\pi+1\right)-\frac{1}{2}\cdot\log\left(\frac{3}{4\delta(1-\delta)}-1\right)\right)\]
 Since $0<f_{\delta}\leq1$ we have that\[
\left|f_{\delta}(v_{1})\log f_{\delta}(v_{1})\cdot\chi_{[-\sqrt{N},\sqrt{N}]}(v_{1})\cdot\left(e^{-\frac{\left(1-v_{1}^{2}\right)^{2}}{(N-1)\Sigma_{\delta}^{2}}}+\sqrt{2\pi}\lambda_{1}\left(N-1,N-v_{1}^{2}\right)\right)\right|\]
 \[
\leq\left(1+\sqrt{2\pi}\epsilon_{1}(N)\right)\left(-f_{\delta}(v_{1})\log f_{\delta}(v_{1})\right)\]
 \[
\leq\left(1+\sqrt{2\pi}\epsilon_{1}(N)\right)\left(-\delta M_{\frac{1}{2\delta}}(v_{1})\log\left(\delta M_{\frac{1}{2\delta}}(v_{1})\right)-(1-\delta)M_{\frac{1}{2(1-\delta)}}(v_{1})\log\left((1-\delta)M_{\frac{1}{2(1-\delta)}}(v_{1})\right)\right)\]
\[
=g_{\delta}(v_{1})\]
 It is easy to check that \[
g_{\delta_{N}}(v)\underset{N\rightarrow0}{\longrightarrow}-M_{\frac{1}{2}}(v)\log M_{\frac{1}{2}}(v)\]
 and \[
\int_{\mathbb{R}}g_{\delta_{N}}(v)dv\underset{N\rightarrow0}{\longrightarrow}-\int_{\mathbb{R}}M_{\frac{1}{2}}(v)\log M_{\frac{1}{2}}(v)dv=\frac{\log\pi}{2}+\frac{1}{2}\]
.

Since\[
f_{\delta_{N}}(v_{1})\log f_{\delta_{N}}(v_{1})\cdot\chi_{[-\sqrt{N},\sqrt{N}]}(v_{1})\cdot\left(e^{-\frac{4\left(1-v_{1}^{2}\right)^{2}\delta_{N}(1-\delta_{N})}{(N-1)\left(3-4\delta_{\nu}(1-\delta_{N})\right)}}+\sqrt{2\pi}\lambda_{1}\left(N-1,N-v_{1}^{2}\right)\right)\]
 \[
\underset{N\rightarrow\infty}{\longrightarrow}M_{\frac{1}{2}}(v_{1})\log M_{\frac{1}{2}}(v_{1})\]
 we conclude that \[
\frac{\int_{\mathbb{S}^{N-1}(\sqrt{N})}F_{N}\log F_{N}d\sigma^{N}}{N}\underset{N\rightarrow\infty}{\longrightarrow}\int_{\mathbb{R}}M_{\frac{1}{2}}(v_{1})\log M_{\frac{1}{2}}(v_{1})dv_{1}+\frac{1}{2}+\frac{\log2\pi}{2}=\frac{\log2}{2}\]
due to the generalized dominated convergence theorem. \end{proof}
\begin{lem}
\label{lem:numerator of the entropy production}Let $F_{N}\left(v_{1},\dots,v_{N}\right)=\frac{\Pi_{i=1}^{N}f_{\delta_{N}}(v_{i})}{Z_{N}(f,\sqrt{N})}$
where $\delta_{N}$ satisfies conditions (\ref{eq:delta conditions}).
Then there exists a constant $C_{type-\delta}$ depending only on
the behavior of $\delta_{N}$ such that \[
\frac{\left\langle \log F_{N},N(I-Q)F_{N}\right\rangle }{N}\leq C_{type-\delta}\left(-\delta_{N}\log\delta_{N}\right)\]
 \end{lem}
\begin{proof}
Similar to Lemma \ref{lem:denominator of entropy production} by using
the symmetry of the problem, Lemma \ref{lem:Integration-on-the-Sphere II}
from the Appendix, Theorem \ref{thm:approximation of Z} and Stirling's
formula we find that \[
\left\langle \log F_{N},N(I-Q)F_{N}\right\rangle \]
 \[
=\frac{1}{Z_{N}(f_{\delta},\sqrt{N})(N-1)\pi}\sum_{k=1}^{N}\int_{\mathbb{S}^{N-1}(\sqrt{N})}\log f_{\delta}(v_{k})\]
 \[
\cdot\left(\sum_{i<j}\int_{0}^{2\pi}\left(f^{\otimes N}\left(v_{1},\dots,v_{N}\right)-f^{\otimes N}\left(R_{i.j}(\vartheta)\left(v_{1},\dots,v_{N}\right)\right)\right)d\vartheta\right)d\sigma^{N}\]
 if $i$ and $j$ are different than $k$ the integral is zero and
so\[
\left\langle \log F_{N},N(I-Q)F_{N}\right\rangle =\frac{1}{Z_{N}(f_{\delta},\sqrt{N})(N-1)\pi}\sum_{k=1}^{N}\sum_{j\not=k}\int_{\mathbb{S}^{N-1}(\sqrt{N})}\log f_{\delta}(v_{k})\]
 \[
\cdot\left(\int_{0}^{2\pi}\left(f^{\otimes N}\left(v_{1},\dots,v_{N}\right)-f^{\otimes N}\left(R_{k.j}(\vartheta)\left(v_{1},\dots,v_{N}\right)\right)\right)d\vartheta\right)d\sigma^{N}\]
 \[
=\frac{N}{Z_{N}(f_{\delta},\sqrt{N})\pi}\int_{0}^{2\pi}d\vartheta\int_{\mathbb{S}^{N-1}(\sqrt{N})}\left(-\log f_{\delta}(v_{1})\right)\left(f_{\delta}(v_{1}(\vartheta))f_{\delta}(v_{2}(\vartheta))-f_{\delta}(v_{1})f_{\delta}(v_{2})\right)\left(\Pi_{i=3}^{N}f_{\delta}(v_{i})\right)d\sigma^{N}\]
 \[
=\frac{N|\mathbb{S}^{N-3}|}{|\mathbb{S}^{N-1}|N^{\frac{N-2}{2}}\pi}\int_{0}^{2\pi}d\vartheta\int_{v_{1}^{2}+v_{2}^{2}\leq N}\left(-\log f_{\delta}(v_{1})\right)\left(f_{\delta}(v_{1}(\vartheta))f_{\delta}(v_{2}(\vartheta))-f_{\delta}(v_{1})f_{\delta}(v_{2})\right)\]
 \[
\cdot\left(N-v_{1}^{2}-v_{2}^{2}\right)^{\frac{N-4}{2}}\frac{Z_{N-2}\left(f_{\delta}.\sqrt{N-v_{1}^{2}-v_{2}^{2}}\right)}{Z_{N}(f_{\delta},\sqrt{N})}dv_{1}dv_{2}\]
 \[
=\frac{N}{\pi\sqrt{1-\frac{2}{N}}}\int_{0}^{2\pi}d\vartheta\int_{v_{1}^{2}+v_{2}^{2}\leq N}\left(-\log f_{\delta}(v_{1})\right)\left(f_{\delta}(v_{1}(\vartheta))f_{\delta}(v_{2}(\vartheta))-f_{\delta}(v_{1})f_{\delta}(v_{2})\right)\]
 \[
\cdot\frac{e^{-\frac{\left(2-v_{1}^{2}-v_{2}^{2}\right)}{(N-2)\Sigma_{\delta}^{2}}}+\sqrt{2\pi}\lambda_{2}\left(N-2.N-v_{1}^{2}-v_{2}^{2}\right)}{1+\sqrt{2\pi}\lambda_{0}(N,N)}dv_{1}dv_{2}\]
 Using rotational symmetry and symmetry in the variables we find that
\[
\left\langle \log F_{N},N(I-Q)F_{N}\right\rangle \]
 \[
=\frac{N}{4\pi\sqrt{1-\frac{2}{N}}}\int_{0}^{2\pi}d\vartheta\int_{v_{1}^{2}+v_{2}^{2}\leq N}\left(\log f_{\delta}(v_{1}(\vartheta))f_{\delta}(v_{2}(\vartheta))-\log f_{\delta}(v_{1})f_{\delta}(v_{2})\right)\]
 \[
\left(f_{\delta}(v_{1}(\vartheta))f_{\delta}(v_{2}(\vartheta))-f_{\delta}(v_{1})f_{\delta}(v_{2})\right)\cdot\frac{e^{-\frac{\left(2-v_{1}^{2}-v_{2}^{2}\right)}{(N-2)\Sigma_{\delta}^{2}}}+\sqrt{2\pi}\lambda_{2}\left(N-2.N-v_{1}^{2}-v_{2}^{2}\right)}{1+\sqrt{2\pi}\lambda_{0}(N,N)}dv_{1}dv_{2}\]
 \[
\leq\frac{N}{4\pi\sqrt{1-\frac{2}{N}}}\int_{0}^{2\pi}d\vartheta\int_{\mathbb{R}^{2}}\left(\log f_{\delta}(v_{1}(\vartheta))f_{\delta}(v_{1}(\vartheta))-\log f_{\delta}(v_{1})f_{\delta}(v_{1})\right)\]
 \[
\cdot\left(f_{\delta}(v_{1}(\vartheta))f_{\delta}(v_{2}(\vartheta))-f_{\delta}(v_{1})f_{\delta}(v_{2})\right)\cdot\frac{1+\sqrt{2\pi}\epsilon_{2}(N)}{1+\sqrt{2\pi}\lambda_{0}(N,N)}dv_{1}dv_{2}\]
 \[
=\frac{N\left(1+\sqrt{2\pi}\epsilon_{2}(N)\right)}{\pi\sqrt{1-\frac{2}{N}}\left(1+\sqrt{2\pi}\lambda_{0}(N,N)\right)}\int_{0}^{2\pi}d\vartheta\int_{\mathbb{R}^{2}}\left(-\log f_{\delta}(v_{1})\right)\left(f_{\delta}(v_{1}(\vartheta))f_{\delta}(v_{2}(\vartheta))-f_{\delta}(v_{1})f_{\delta}(v_{2})\right)dv_{1}dv_{2}\]
 Since $M_{a}(v_{1}(\vartheta))M_{a}(v_{2}(\vartheta))=M_{a}(v_{1})M_{a}(v_{2})$
we see that\[
f_{\delta}(v_{1}(\vartheta))f_{\delta}(v_{2}(\vartheta))-f_{\delta}(v_{1})f_{\delta}(v_{2})=\delta(1-\delta)\left(M_{\frac{1}{2\delta}}(v_{1}(\vartheta))M_{\frac{1}{2(1-\delta)}}(v_{2}(\vartheta))-M_{\frac{1}{2\delta}}(v_{1})M_{\frac{1}{2(1-\delta)}}(v_{2})\right)\]
 \[
+\delta(1-\delta)\left(M_{\frac{1}{2\delta}}(v_{2}(\vartheta))M_{\frac{1}{2(1-\delta)}}(v_{1}(\vartheta))-M_{\frac{1}{2\delta}}(v_{2})M_{\frac{1}{2(1-\delta)}}(v_{1})\right)\]
 \[
\leq\delta(1-\delta)\left(M_{\frac{1}{2\delta}}(v_{1}(\vartheta))M_{\frac{1}{2(1-\delta)}}(v_{2}(\vartheta))+M_{\frac{1}{2\delta}}(v_{2}(\vartheta))M_{\frac{1}{2(1-\delta)}}(v_{1}(\vartheta))\right)\]
 and along with \[
-\log f_{\delta}(v_{1})\leq-\log\left(\delta M_{\frac{1}{2\delta}}(v_{1})\right)\leq-\frac{3\log\delta}{2}+\frac{\log\pi}{2}+\delta\left(v_{1}^{2}(\vartheta)+v_{2}^{2}(\vartheta)\right)\]
 we conclude that \[
\frac{\left\langle \log F_{N},N(I-Q)F_{N}\right\rangle }{N}\]
 \[
\leq\frac{4\left(1+\sqrt{2\pi}\epsilon_{2}(N)\right)\delta(1-\delta)}{\sqrt{1-\frac{2}{N}}\left(1+\sqrt{2\pi}\lambda_{0}(N,N)\right)}\int_{\mathbb{R}^{2}}\left(-\frac{3\log\delta}{2}+\frac{\log\pi}{2}+\delta\left(v_{1}^{2}+v_{2}^{2}\right)\right)M_{\frac{1}{2\delta}}(v_{1})M_{\frac{1}{2(1-\delta)}}(v_{2})dv_{1}dv_{2}\]
 \[
\leq\frac{4\left(1+\sqrt{2\pi}\epsilon_{2}(N)\right)}{\sqrt{1-\frac{2}{N}}\left(1+\sqrt{2\pi}\lambda_{0}(N,N)\right)}\left(\frac{3}{2}-\frac{\log\pi}{2\log\delta}-\frac{1}{2\log\delta}-\frac{\delta}{2\log\delta}\right)\left(-\delta\log\delta\right)\]
 The result follows. \end{proof}
\begin{thm}
\label{thm:big result general}Let $F_{N}\left(v_{1},\dots,v_{N}\right)=\frac{\Pi_{i=1}^{N}f_{\delta_{N}}(v_{i})}{Z_{N}(f,\sqrt{N})}$
where $\delta_{N}$ satisfies conditions (\ref{eq:delta conditions}).
Then there exists a constant $C_{type-\delta}$ and an integer $N_{type-\delta}$
depending only on the behavior of $\delta_{N}$ such that for every
$N>N_{type-\delta}$ \[
\frac{\left\langle \log F_{N},N(I-Q)F_{N}\right\rangle }{\int_{\mathbb{S}^{N-1}(\sqrt{N})}F_{N}\log F_{n}d\sigma^{N}}\leq C_{type-\delta}\left(-\delta_{N}\log\delta_{N}\right)\]
 \end{thm}
\begin{proof}
This follows immediately from Lemma \ref{lem:denominator of entropy production}
and \ref{lem:numerator of the entropy production}.\end{proof}
\begin{thm}
\label{thm:big result beta}Let $F_{N}\left(v_{1},\dots,v_{N}\right)=\frac{\Pi_{i=1}^{N}f_{\delta_{N}}(v_{i})}{Z_{N}(f,\sqrt{N})}$
where $\delta_{N}=\frac{1}{N^{1-2\beta}}$ and $0<\beta<\frac{1}{6}$.
Then there exists a constant $C_{\beta}$ and an integer $N_{\beta}$
depending only on $\beta$ such that for every $N>N_{\beta}$ \[
\frac{\left\langle \log F_{N},N(I-Q)F_{N}\right\rangle }{\int_{\mathbb{S}^{N-1}(\sqrt{N})}F_{N}\log F_{n}d\sigma^{N}}\leq\frac{C_{\beta}\log N}{N^{1-2\beta}}\]
 \end{thm}
\begin{proof}
This follows immediately from Theorem \ref{thm:big result general}
and the fact that $\delta_{N}=\frac{1}{N^{1-2\beta}}$ satisfies conditions
(\ref{eq:delta conditions}).

From this we conclude our main result:\end{proof}
\begin{thm}
\label{thm:tnropy production result}For any $0<\beta<\frac{1}{6}$
there exists a constant $C_{\beta}$ depending only on $\beta$ such
that\[
\Gamma_{N}\leq\frac{C_{\beta}\log N}{N^{1-2\beta}}\]

\end{thm}

\section{Final Remarks\label{sec:Final-Remarks}}

One question we might ask ourselves is: Can we modify the given proof
to get the exact value in Villani's conjecture? Looking at the proof
we notice that the result we obtained has very tight conditions in
terms of $\beta$. We needed $\delta_{N}^{1+2\beta}N$ to diverge
to infinity \emph{and} $\delta_{N}^{1+3\beta}N$ to go to zero. This
doesn't leave much room for variations. This leads us to believe that
the family of functions constructed here would not be helpful to prove
the exact version of Villani's conjecture. Something more clever must
be done.

Another question we don't know the answer to is the fourth moment
question. Both in this paper and in (\cite{key-8}) the family of
functions constructed has an unbounded fourth moment. Would restricting
the fourth moment lead to a lower bound on the entropy production?

Lastly, can our computation be generalized to a more difficult interaction
than Kac's model? Can we try and use the same idea in a different
models of the Boltzmann equation?

While we don't know the answers to the proposed questions we hope
that this paper shed some light on the entropy production problem
and that at least some of the above questions would seem more solvable
after reading it.

\appendix

\section{Helpful Computations}

The appendix consists of Lemmas that are vital for the computations
needed in our paper, and are used extensively in Sections \ref{sec:Centrel-Limit-Theorem}
and \ref{sec:Entropy-Production}.
\begin{lem}
\label{lem:Gaussian-integral-estimation}(Gaussian Integral Estimation)\[
\frac{\sqrt{2\pi}}{a}\cdot\sqrt{1-e^{-\frac{a\eta^{2}}{2}}}\leq\int_{|x|<\eta}e^{-\frac{a^{2}x^{2}}{2}}dx\leq\frac{\sqrt{2\pi}}{a}\cdot\sqrt{1-e^{-a^{2}\eta^{2}}}\]
 \[
\int_{|x|>\eta}e^{-\frac{a^{2}x^{2}}{2}}dx\leq\frac{\sqrt{2\pi}\cdot e^{-\frac{a^{2}\eta^{2}}{2}}}{a}\]
 \end{lem}
\begin{proof}
We have\[
\int_{|x|<\eta}e^{-\frac{a^{2}x^{2}}{2}}dx=\sqrt{\int\int_{|x|,|y|<\eta}e^{-\frac{a^{2}\left(x^{2}+y^{2}\right)}{2}}dxdy}\leq\sqrt{\int\int_{x^{2}+y^{2}<2\eta^{2}}e^{-\frac{a^{2}\left(x^{2}+y^{2}\right)}{2}}dxdy}\]
 \[
=\sqrt{\int_{0}^{2\pi}\int_{0}^{\sqrt{2}\eta}re^{-\frac{a^{2}r^{2}}{2}}drd\vartheta}=\sqrt{2\pi}\cdot\sqrt{\frac{1-e^{-a^{2}\eta^{2}}}{a^{2}}}\]
 And \[
\int_{|x|<\eta}e^{-\frac{a^{2}x^{2}}{2}}dx\geq\sqrt{\int\int_{x^{2}+y^{2}<\eta^{2}}e^{-\frac{a^{2}\left(x^{2}+y^{2}\right)}{2}}dxdy}=\sqrt{2\pi}\cdot\sqrt{\frac{1-e^{-\frac{a\eta^{2}}{2}}}{a^{2}}}\]
 Similarly\[
\int_{|x|>\eta}e^{-\frac{a^{2}x^{2}}{2}}dx=\int_{\mathbb{R}}e^{-\frac{a^{2}x^{2}}{2}}dx-\int_{|x|<\eta}e^{-\frac{a^{2}x^{2}}{2}}dx=\frac{\sqrt{2\pi}}{a}-\int_{|x|<\eta}e^{-\frac{a^{2}x^{2}}{2}}dx\]
 \[
\leq\frac{\sqrt{2\pi}}{a}\left(1-\sqrt{1-e^{-\frac{a^{2}\eta^{2}}{2}}}\right)=\frac{\sqrt{2\pi}\cdot e^{-\frac{a^{2}\eta^{2}}{2}}}{a\left(1+\sqrt{1-e^{-\frac{a^{2}\eta^{2}}{2}}}\right)}\leq\frac{\sqrt{2\pi}\cdot e^{-\frac{a^{2}\eta^{2}}{2}}}{a}\]
 \end{proof}
\begin{lem}
\label{lem:Special-Sums-Evaluation}(Special Sums Evaluation)\[
\sum_{k=k_{0}+1}^{m}\frac{e^{-\frac{a^{2}k}{2}}}{\sqrt{k}}\leq\frac{\sqrt{2\pi}\cdot e^{-\frac{a^{2}k_{0}}{2}}}{a}\]
 \[
\sum_{k=k_{0}+1}^{m}\frac{1}{\sqrt{k}}\leq2\sqrt{m}\]
 \end{lem}
\begin{proof}
We have that\[
\sum_{k=k_{0}+1}^{m}\frac{e^{-\frac{a^{2}k}{2}}}{\sqrt{k}}\leq\int_{k_{0}}^{m}\frac{e^{-\frac{a^{2}x}{2}}}{\sqrt{x}}dx\underset{y=a\sqrt{x}}{=}\frac{2}{a}\int_{a\sqrt{k_{0}}}^{a\sqrt{m}}e^{-\frac{y^{2}}{2}}dy\leq\frac{2}{a}\int_{a\sqrt{k_{0}}}^{\infty}e^{-\frac{y^{2}}{2}}dy\]
 \[
=\frac{1}{a}\int_{|y|>a\sqrt{k_{0}}}e^{-\frac{y^{2}}{2}}dy\leq\frac{\sqrt{2\pi}\cdot e^{-\frac{a^{2}k_{0}}{2}}}{a}\]
 Similarly\[
\sum_{k=k_{0}+1}^{m}\frac{1}{\sqrt{k}}\leq\int_{k_{0}}^{m}\frac{dx}{\sqrt{x}}=2\left(\sqrt{m}-\sqrt{k_{0}}\right)\leq2\sqrt{m}\]

\end{proof}
The next set of Lemmas refer to integration over the sphere $\mathbb{S}^{N-1}(r)$. 
\begin{lem}
\label{lem:Integration-on-the-Sphere I}(Integration on the Sphere
I) Let $f\left(v_{1},\dots,v_{N}\right)$ be a continuous function
on $\mathbb{R}^{N}$ then\[
\int_{\mathbb{S}^{N-1}(r)}fds_{r}^{N}=\sum_{\epsilon=\left\{ +,-\right\} }\int_{\sum_{i=1}^{N-1}v_{i}^{2}\leq r^{2}}\frac{r\cdot f\left(v_{1},\dots,v_{N-1},\epsilon\sqrt{r^{2}-\sum_{i=1}^{N-1}v_{i}^{2}}\right)}{\sqrt{r^{2}-\sum_{i=1}^{N-1}v_{i}^{2}}}dv_{1}\dots dv_{N-1}\]
 \end{lem}
\begin{proof}
Standard in any Differential Geometry course.\end{proof}
\begin{cor}
\label{cor:Integration-on-the-Sphere}(Integration on the Sphere with
the Uniform Probability Measure)\[
\int_{\mathbb{S}^{N-1}(r)}fd\sigma_{r}^{N}=\frac{1}{|\mathbb{S}^{N-1}|r^{N-2}}\cdot\sum_{\epsilon=\left\{ +,-\right\} }\int_{\sum_{i=1}^{N-1}v_{i}^{2}\leq r^{2}}\frac{f\left(v_{1},\dots,v_{N-1},\epsilon\sqrt{r^{2}-\sum_{i=1}^{N-1}v_{i}^{2}}\right)}{\sqrt{r^{2}-\sum_{i=1}^{N-1}v_{i}^{2}}}dv_{1}\dots dv_{N-1}\]
 \end{cor}
\begin{lem}
\label{lem:Integration-on-the-Sphere II}(Integration on the Sphere
II) Let $f\left(v_{1},\dots,v_{j}\right)$ and $g\left(v_{j+1},\dots,v_{N}\right)$
be continuous functions on $\mathbb{R}^{j}$ and $\mathbb{R}^{N-j}$
respectfully. Then\[
\int_{\mathbb{S}^{N-1}(r)}f\left(v_{1},\dots,v_{j}\right)\cdot g\left(v_{j+1},\dots,v_{N}\right)d\sigma_{r}^{N}\]
 \[
=\frac{|\mathbb{S}^{N-j-1}|}{|\mathbb{S}^{N-1}|r^{N-2}}\int_{\sum_{i=1}^{j}v_{i}^{2}\leq r^{2}}f\left(v_{1},\dots,v_{j}\right)\left(r^{2}-\sum_{i=1}^{j}v_{i}^{2}\right)^{\frac{N-j-2}{2}}\]
 \[
\left(\int_{\mathbb{S}^{N-j-1}\left(\sqrt{r^{2}-\sum_{i=1}^{j}v_{i}^{2}}\right)}gd\sigma_{\sqrt{r^{2}-\sum_{i=1}^{j}v_{i}^{2}}}^{N-j}\right)dv_{1}\dots dv_{j}\]
 \end{lem}
\begin{proof}
Using Corollary \ref{cor:Integration-on-the-Sphere} we find that\[
\int_{\mathbb{S}^{N-1}(r)}f\left(v_{1},\dots,v_{j}\right)\cdot g\left(v_{j+1},\dots,v_{N}\right)d\sigma_{r}^{N}\]
 \[
\frac{\sum_{\epsilon=\left\{ +,-\right\} }}{|\mathbb{S}^{N-1}|r^{N-2}}\int_{\sum_{i=1}^{N-1}v_{i}^{2}\leq r^{2}}\frac{f\left(v_{1},\dots,v_{j}\right)\cdot g\left(v_{j+1},\dots,v_{N-1},\epsilon\sqrt{r^{2}-\sum_{i=1}^{N-1}v_{i}^{2}}\right)}{\sqrt{r^{2}-\sum_{i=1}^{N-1}v_{i}^{2}}}dv_{1}\dots dv_{N-1}\]
 \[
=\frac{1}{|\mathbb{S}^{N-1}|r^{N-2}}\int_{\sum_{i=1}^{j}v_{i}^{2}\leq r^{2}}\frac{f\left(v_{1},\dots,v_{j}\right)}{\sqrt{r^{2}-\sum_{i=1}^{j}v_{i}^{2}}}\left(\int_{\mathbb{S}^{N-j-1}\left(\sqrt{r^{2}-\sum_{i=1}^{j}v_{i}^{2}}\right)}gds_{\sqrt{r^{2}-\sum_{i=1}^{j}v_{i}^{2}}}^{N-j}\right)dv_{1}\dots dv_{j}\]
 \[
=\frac{|\mathbb{S}^{N-j-1}|}{|\mathbb{S}^{N-1}|r^{N-2}}\int_{\sum_{i=1}^{j}v_{i}^{2}\leq r}f\left(v_{1},\dots,v_{j}\right)\left(r^{2}-\sum_{i=1}^{j}v_{i}^{2}\right)^{\frac{N-j-2}{2}}\]
 \[
\left(\int_{\mathbb{S}^{N-j-1}\left(\sqrt{r^{2}-\sum_{i=1}^{j}v_{i}^{2}}\right)}gd\sigma_{\sqrt{r^{2}-\sum_{i=1}^{j}v_{i}^{2}}}^{N-j}\right)dv_{1}\dots dv_{j}\]
\end{proof}

Author's email: aeinav@math.gatech.edu
\end{document}